\documentclass[conference]{IEEEtran}
\IEEEoverridecommandlockouts

\usepackage{cite}
\usepackage{graphicx,enumerate}
\usepackage{caption}
\usepackage{subcaption}
\usepackage{graphicx}
\usepackage{mathtools,commath,amsmath,amssymb,amsfonts,bm}\interdisplaylinepenalty=2500
\usepackage{amsthm}
\usepackage{algpseudocode,algorithm}
\usepackage{tabularx,adjustbox}
\usepackage{todonotes}
\usepackage{breqn}
\clearpage
\extrafloats{100}
\RequirePackage{tcolorbox}
\tcbuselibrary{breakable}
\tcbset{parbox=false, left=1ex, right=1ex}

\usepackage{multirow}
\usepackage{xspace}
\usepackage{url}
\newcolumntype{C}[1]{>{\hsize=#1\hsize\centering\arraybackslash}X}

\def\vec#1{\mathchoice{\mbox{\boldmath$\displaystyle#1$}}
  {\mbox{\boldmath$\textstyle#1$}}
  {\mbox{\boldmath$\scriptstyle#1$}}
  {\mbox{\boldmath$\scriptscriptstyle#1$}}}
\newcommand{\mat}[1]{\mathbf{#1}}

\newcommand{\definedas}{\overset{\mathrm{def}}{=}}
\newcommand{\etal}{\textit{et~al.}\xspace}
\newcommand{\ie}{\textit{i.e.},\xspace}

\newcommand{\mech}{\ensuremath{\mathcal{M}}\xspace}
\newtheorem{definition}{Definition}[section]

\newtheorem{theorem}{Theorem}
\newtheorem{corollary}{Corollary}

\newtheorem{dfn}{Definition}

\algnewcommand\algorithmicinput{\textbf{Input:}}
\algnewcommand\Input{\item[\algorithmicinput]}
\algnewcommand\algorithmicoutput{\textbf{Output:}}
\algnewcommand\Output{\item[\algorithmicoutput]}
\algnewcommand\algorithmictier{\textbf{Step:}}
\algnewcommand\Tier{\item[\algorithmictier]}

\floatname{algorithm}{Protocol}
\def\BibTeX{{\rm B\kern-.05em{\sc i\kern-.025em b}\kern-.08em
    T\kern-.1667em\lower.7ex\hbox{E}\kern-.125emX}}

\begin{document}
\title{Towards Distributed Privacy-Preserving Prediction}

\author{\IEEEauthorblockN{1\textsuperscript{st} Lingjuan~Lyu}
\IEEEauthorblockA{\textit{Department of Computer Science} \\
\textit{National University of Singapore} \\
lyulj@comp.nus.edu.sg}
\and
\IEEEauthorblockN{2\textsuperscript{nd} Yee Wei~Law}
\IEEEauthorblockA{\textit{School of Engineering} \\
\textit{University of South Australia} \\
yeewei.law@unisa.edu.au
}
\and
\IEEEauthorblockN{3\textsuperscript{rd} Kee Siong~Ng}
\IEEEauthorblockA{\textit{Software Innovation Institute} \\
\textit{Australian National University}\\
keesiong.ng@anu.edu.au}
\and
\IEEEauthorblockN{4\textsuperscript{th} Shibei~Xue}
\IEEEauthorblockA{\textit{Department of Automation,} \\ 
\textit{Shanghai Jiao Tong University} \\ 
\textit{Key Laboratory of System Control} \\
\textit{and Information Processing,} \\
\textit{Ministry of Education of China} \\
 shbxue@sjtu.edu.cn
}
\and
\IEEEauthorblockN{5\textsuperscript{th} Jun Zhao, 6\textsuperscript{th} Mengmeng Yang}
\IEEEauthorblockA{\textit{School of Computer Science and Engineering} \\
\textit{Nanyang Technological University}\\
Singapore \\
junzhao, melody.yang@ntu.edu.sg }
\and
\IEEEauthorblockN{7\textsuperscript{th} Lei Liu}
\IEEEauthorblockA{\textit{Unicloud Engine Technology Co., Ltd} \\
China \\
liulei@unicde.com}
\thanks{Corresponding to: shbxue@sjtu.edu.cn. This work was supported in part by the National Natural Science Foundation of China under Grants 61873162 and in part by the Open Research Project of the State Key Laboratory of Industrial Control Technology, Zhejiang University, China (No.ICT20052).}
}

\maketitle
\begin{sloppypar}
\begin{abstract}
In privacy-preserving machine learning, individual parties are reluctant to share their sensitive training data due to privacy concerns. Even the trained model parameters or prediction can pose serious privacy leakage. To address these problems, we demonstrate a generally applicable \emph{Distributed Privacy-Preserving Prediction} (DPPP) framework, in which instead of sharing more sensitive data or model parameters, an untrusted aggregator combines only multiple models' predictions under provable privacy guarantee.
Our framework integrates two main techniques to guarantee individual privacy. First, we introduce the improved Binomial Mechanism and Discrete Gaussian Mechanism to achieve distributed differential privacy.
Second, we utilize homomorphic encryption to ensure that the aggregator learns nothing but the noisy aggregated prediction. Experimental results demonstrate that our framework has comparable performance to the non-private frameworks and delivers better results than the local differentially private framework and standalone framework.
\end{abstract}


\section{Introduction}
Many real-world applications would benefit from collaborative learning among multiple parties. This trend is motivated by the fact that the data owned by a single party may be very heterogeneous, resulting in an overfit model that might deliver inaccurate results when applied to other data. On the other hand, there is much demand to perform machine learning in a collaborative manner, since massive amount of data are often required to ensure sufficient computational power for test purpose. However, the increasing privacy and confidentiality concerns pose obstacles to collaboration~\cite{ohno2004protecting}. For the sake of privacy, most approaches cannot afford to share the trained models publicly. Even the prediction output by a trained model can reveal training data privacy through black-box attacks~\cite{shokri2017membership}. Therefore, neither training data, trained model nor model prediction should be directly shared. Meanwhile, these privacy concerns can be largely reduced if appropriate privacy-preserving schemes can be applied before the relevant statistic is released.

To mitigate privacy concerns in the distributed setting, instead of sharing more sensitive local data or model parameters of each party, we examine an alternative approach called \emph{distributed privacy-preserving prediction} (DPPP), which allows parties to keep full control of their local data, and only share local model predictions in a privacy-preserving manner. In our approach, each party first trains a local model based on local training data $\mat{D}_i$. To answer the prediction query for any test point $\vec{x}$, each party takes local model as the prediction function $f$ to predict $\vec{x}$, which returns the votes for all $c$ classes, \ie $f(\vec{x},\mat{D}_i)=\vec{y}_i$, where $\vec{y}_i \in \{0, 1\}^c$ is an one-hot prediction vector that sums up to 1. Considering individual privacy, we appropriately perturb model predictions before releasing them for aggregation. The primitive of this provably private noisy sum can be referred to~\cite{blum2005practical}. \textbf{Our contributions} include: 

\begin{itemize}
\item We formulate a distributed privacy-preserving prediction framework, named DPPP, which combines \emph{distributed differential privacy} (DDP) and homomorphic encryption to ensure individual privacy, maintain utility and provide aggregator obliviousness, \ie the aggregator learns nothing but the noisy aggregated prediction.

\item We explore the stability of \emph{Binomial Mechanism} (\emph{BM}) and \emph{Discrete Gaussian Mechanism} (\emph{DGM}) to guarantee ($\epsilon, \delta$)-differential privacy in the distributed setting, and formally provide a tightest bound to date for \emph{BM}. 

\item Extensive experiments demonstrate that DPPP delivers comparable performance to the \mbox{non-private} frameworks, and yields better results than the local differentially private and standalone frameworks.
\end{itemize}

\section{Preliminaries and Related Work}
\label{sec:PRELIMINARIES}

\subsection{Distributed Differential Privacy}

\begin{dfn}[$(\epsilon,\delta,\gamma)$-Distributed DP~\cite{shi2011privacy}] \label{def:ddp_database} Let $\epsilon>0$, $0\leq\delta<1$ and $0<\gamma<1$, We say that the mechanism \mech with randomness over the joint distribution of $\vec{r}:=(\vec{r_1},\cdots,\vec{r_N})$ preserves $(\epsilon,\delta,\gamma)$-distributed differential privacy (DDP) if the following conditions hold: For any neighbouring databases $D,D'\in\mathcal{D}^N$ that differ in one record, for any measurable subset $S\subseteq\mathcal{R}$, and for any subset $\bar{K}$ of at least $\gamma N$ honest parties, 
\begin{eqnarray*}
\label{eq:ddp_database}
\Pr\{\mech(D)\in S|\vec{r}_K\} &\leq& \exp(\epsilon) \cdot \Pr\{\mech(D')\in S|\vec{r}_K\} +
\delta.
\end{eqnarray*}
\end{dfn}

In the above definition, $\gamma$ is the fraction of uncompromised parties, and the probability is conditioned on the randomness $\vec{r}_K$ from compromised parties, \ie it ensures that if at least $\gamma N$ participants are honest and uncompromised, we will accumulate noise of a similar magnitude as that of the \emph{central differential privacy} (CDP)~\cite{dwork2014algorithmic}. For differentially-private aggregation of local statistics, DDP permits each party to randomise its local statistic to a lesser degree than \emph{local differential privacy} (LDP)~\cite{lyu2017privacy,yang2020local}. The most recent work introduced amplification by shuffling to lower the privacy cost of LDP algorithm when viewed in the central model of DP~\cite{erlingsson2019amplification,balle2019privacy}. LDP with shuffling yields a trust model which sits in between the central and local models for DP. 

\subsection{Homomorphic Encryption} 
Additive homomorphic encryption allows the calculation of the encrypted sum of plaintexts from their corresponding ciphertexts. 
Although there are several additive homomorphic cryptographic schemes, we use the threshold variant of Paillier scheme~\cite{damgaard2001generalisation} in our framework, because it not only allows additive homomorphic encryption, but also distributes decryption among parties.
In this cryptosystem, a party can encrypt the plaintext $m \in \mathbb{Z}_n$ with the public key $pk=(g,n)$ as

\begin{equation}\label{eq:enc}
c=E_{pk}(m)=g^mr^n \mod n^2,
\end{equation}
where $r \in \mathbb{Z}_n^*$ ($\mathbb{Z}_n^*$ denotes the multiplicative group of invertible elements of $\mathbb{Z}_n$) is selected randomly and privately by each party. The additive homomorphic property of this cryptosystem can be described as:

\begin{equation}\label{eq:add_hor}
\begin{split}
E_{pk}(m_1+m_2) &= E_{pk}(m_1) \cdot E_{pk}(m_2) \\
                &= g^{m_1+m_2}(r_1r_2)^n \mod n^2,
\end{split}
\end{equation}
where $m_1$, $m_2$ are the plaintexts that need to be encrypted, and $r_1$, $r_2$ are the private randoms.

In this paper, $(N, t)$-threshold Paillier cryptosystem is adopted, in which the private key $sk$ is distributed among $N$ parties (denoted as $\{sk_1, sk_2, \cdots, sk_N\}$), thus no single party has the complete private key. For any ciphertext $c$, each party $i$ ($1 \leq i \leq N$) computes a partial decryption with its own partial private key $sk_i$ as:
\begin{equation}\label{eq:partial_dec}
c_i =c^{2N!sk_i}
\end{equation}
Then based on the combining algorithm in~\cite{damgaard2001generalisation}, at least $t$ partial decryptions are required to recover the plaintext $m$. 

\subsection{Multi-party Privacy}
In multi-party scenario where data is sourced from multiple parties and the server is \textbf{not trustworthy}, individual privacy has to be protected. 
Without homomorphic encryption, each party has to add sufficient noise to their statistics before sending them to a central server to ensure LDP~\cite{yang2020local}. Since the aggregation sums up individual noise shares, the aggregated noise might render the aggregation useless. To preserve privacy without significantly degrading utility, differential privacy can be made distributed by combining with cryptographic protocols, as evidenced in~\cite{rastogi2010differentially,acs2011have,shi2011privacy,lyu2017privacy}. 

More recently, Agarwal \etal~\cite{agarwal2018cpsgd} recalled the Binomial Mechanism~\cite{dwork2006our} and provided a similar bound as ours in Theorem~\ref{Theorem_BM}. However, they focus on the privacy of the gradients aggregated from clients in federated learning, which is different from the problem studied in this work. More importantly, they did not offer a complete scheme to protect against the untrusted aggregator. Note that they did not provide a complete proof. Instead, we provide a detailed proof for the tight bound.

Another recent work is \emph{Private Aggregation of Teacher Ensembles} (PATE) proposed by Papernot \etal~\cite{papernot2017semi}. PATE first trains an ensemble of teachers on disjoint subsets of private data. These teachers are then used to train a student model that can accurately mimic the ensemble. However, PATE assumes a trusted aggregator, who counts teacher votes assigned to each class, adds carefully calibrated Laplace noise to the resulting vote histogram, and outputs the class with the most noisy votes as the ensemble's prediction. Therefore, PATE fails to take into consideration against a potentially untrusted aggregator.

\section{Problem Definition}
Similar to Shi \etal~\cite{shi2011privacy}, we consider an untrusted aggregator who may have arbitrary auxiliary information. For example, the aggregator may collude with a set of compromised parties, who can \emph{reveal their data and noise values} to the aggregator as a form of auxiliary information. Our goal is to guarantee the privacy of each individual against an untrusted aggregator, even when the aggregator has arbitrary auxiliary information. To achieve this goal, we blind and encrypt the local statistics of parties before sharing them with the aggregator. Moreover, like most of the previous works~\cite{rastogi2010differentially,acs2011have,shi2011privacy}, to ensure the correctness and functionality of the system, we do not consider a malicious aggregator, as it may not be desirable in many practical settings, and are not in the commercial interest of collaborative service providers for prediction service. We remark that our privacy model is stronger than~\cite{shi2011privacy} in the sense that we allow for party failures. This assumption is often more realistic, as one or more parties may fail to upload their encrypted values or fail to respond. 

Moreover, we assume fewer than $1-\gamma=1/3$ of teachers are compromised -- the rest are assumed to be honest. Decryption can be done by the remaining 2/3 teachers using threshold Paillier. Since cryptographic protocol requires discrete inputs, instead of adding floating-point Gaussian noise to each individual's prediction, we leverage Binomial Mechanism (\emph{BM}) and Discrete Gaussian Mechanism (\emph{DGM}) to generate discrete Binomial noise and Gaussian noise. 

\section{DDP Mechanisms}
\label{sec:DDP} 
Approaches to DDP that implement an overall additive noise mechanism by summing the same mechanism run at each party (typically with less noise) necessitates mechanisms with stable distributions---to guarantee proper calibration of known end-to-end response distribution---and cryptography for hiding all but the final result from participating parties~\cite{shi2011privacy,dwork2006our,acs2011have,rastogi2010differentially}.
DDP utilizes this nice stability to permit each party to randomise its local statistic to a lesser degree ($\frac{\sigma}{\sqrt{n}}$) than would LDP ($\sigma$)~\cite{lyu2017privacy}. In summary, the goal of DDP is to both avoid the trust on any third party (trusted server in CDP~\cite{dwork2014algorithmic} and trusted shuffler in LDP with shuffling~\cite{yang2020local}), and achieve better utility than LDP. On the other hand, if the server colludes with all the parties except the victim, the privacy guarantee would downgrade to LDP. We next introduce two representative stable distributions, including Binomial distribution and discrete Gaussian distribution, which can be seamlessly combined with cryptographic techniques.

\subsection{Binomial Mechanism}
\emph{Binomial Mechanism} (\emph{BM}) is based on the Binomial distribution $B(n, p)$ parameterized by $n$, $p$, where $n\in \mathbb{N}$ is the number of tosses, and $p \in (0, 1)$ is the success probability. 
We now define \emph{BM} for the prediction function $f$ with an output space in $\{0, 1\}^c$, where $c$ is the number of classes, \ie for each data point, $f$ returns an one-hot prediction vector with a total of $c$ elements in $\{0, 1\}$ that sum up to 1.
Consider party $i$'s database $\mat{D}_i$ and prediction $f$: let $f(\vec{x},\mat{D}_i)=\vec{y}_i$ be the local prediction vector produced by party $i$ given data $\mat{D}_i$, and $y_i^j$ be the $j$-th element of $\vec{y}_i$, \ie prediction (vote status) for class $j$. If party $i$ assigns class $j$ to input $\vec{x}$, then $y_i^j=1$ while other elements are all 0's. The noisy vote count $\tau$ on each class $j$ equals $\tau= \textstyle\sum_{i=1}^N y_i^j + \textbf{noise}$, replacing the whole database $\mat{D}$ with $\mat{D}'$ differing only in one row changes the summation in each class by at most 1, \ie sensitivity=1. Bounding the ratio of probabilities that $\tau$ occurs with inputs $\mat{D}$ and $\mat{D}'$ amounts to bounding the ratio of probabilities that \textbf{noise} = $r^j$ and \textbf{noise} = $r^j$ + 1, for different possible ranges of values of $r^j$. Given prediction function $f(\vec{x},\mat{D}_i)=\vec{y}_i$, the goal of the \emph{BM} is to compute the noisy vote count for each class (each coordinate of the aggregated prediction): $\textstyle\sum_{i=1}^N y_i^j+r^j$, where $r^j=z-np$ is the random Binomial noise added to the vote count for each class $j$, and Binomial random variable $z \sim B(n, p)$ is independent for each class.

\begin{theorem}
\label{Theorem_BM}
(Tighter bound). 
For $p=1/2$, Binomial Mechanism is $(\epsilon,\delta)$-differentially private so long as the total number of tosses $n \geq 2\left(\frac{2+\epsilon}{\epsilon}\right)^2\ln\left(\frac{2}{\delta}\right)$. Note that this lower bound is tighter than $n \geq 64\ln\left(\frac{2}{\delta}\right)/\epsilon^2$ given in~\cite{dwork2006our}, but they both share the $\ln\left(\frac{2}{\delta}\right)$ term.
\end{theorem}

\begin{proof}
For Binomial distribution with $p=\frac{1}{2}$, $r=z-\frac{n}{2}$ is termed as Binomial noise, where $z$ is a Binomial random variable sampled from $B(n,1/2)$ with mean $\frac{n}{2}$ and success probability $\frac{1}{2}$ by performing coin flipping. To investigate how we can size Binomial noise, suppose $r^j$ is the random Binomial noise added to the vote count for each class $j$, then
\[ \tau = \textstyle\sum_{i=1}^N y_i^j + r^j, \qquad \tau' = \textstyle\sum_{i=1}^N y_i^j + r^j + 1. \]

For the Binomial random variable with bias $1/2$, whose mass at $\frac{n}{2} + r^j$ is
\[\Pr\left(\frac{n}{2}+r^j\right)=\binom{n}{\frac{n}{2}+r^j}\left(\frac{1}{2}\right)^n, \]
$\epsilon$-differential privacy requires that
\[\begin{split}
&\Pr\left\{\frac{n}{2} + r^j\right\} \leq e^\epsilon\Pr\left\{\frac{n}{2} + r^j + 1\right\} \\
	\implies
	&\binom{n}{\frac{n}{2}+r^j}\left(\frac{1}{2}\right)^n \leq e^\epsilon\binom{n}{\frac{n}{2}+r^j+1}\left(\frac{1}{2}\right)^n \\
	\implies
	&\frac{n}{2}+r^j+1 \leq e^\epsilon\left(\frac{n}{2}-r^j\right).
\end{split}\]

To express $r^j$ in terms of $\epsilon$ in an algebraically simple way, we use the inequality $1+\epsilon \leq e^\epsilon$:
\begin{equation*}
\begin{split}
	\frac{n}{2}+r^j+1 \leq (1+\epsilon)\left(\frac{n}{2}-r^j\right) \implies	r^j \leq \frac{\epsilon n - 2}{4 + 2\epsilon} < \frac{\epsilon n}{4 + 2\epsilon}.
\end{split}
\end{equation*}
Therefore, the Binomial random variable $n/2+r^j$ can \emph{apparently} achieve $\epsilon$-differential privacy, as long as $r^j<\frac{\epsilon n}{4 + 2\epsilon}$. Note:

\begin{itemize}
	\item $\frac{n}{2}+\frac{\epsilon n}{4 + 2\epsilon} = \left(1 + \frac{\epsilon}{2+\epsilon}\right)\frac{n}{2}$. This equation will be used in Eq.~\eqref{eq:binom-Chernoff} later.
	\item This noise upper bound is tighter than Dwork \etal's \cite{dwork2006our}. 
\end{itemize}
However, note that when $r^j$ exceeds this upper bound, $\epsilon$-differential privacy will be violated. Hence, we turn to the relaxed $(\epsilon,\delta)$-DP, which requires that
\[\begin{split}
	\Pr\left\{\frac{n}{2}+r^j\leq \frac{n}{2}-\xi\text{ or } \frac{n}{2}+r^j\geq \frac{n}{2}+\xi\right\} \leq \delta,
\end{split}\]
where $\xi\definedas\frac{\epsilon n}{4 + 2\epsilon}$. Since the Binomial distribution is symmetrical about its mean $n/2$, the inequality above is equivalent to 
\[
	\Pr\cbr{\frac{n}{2}+r^j \geq \frac{n}{2}+\xi} \leq \frac{\delta}{2},
\]
According to Chernoff bound theorem, for any $X \sim$ Binomial$(n, 1/2)$, and $0 \leq t \leq \sqrt{n}$,
\[ \Pr\left\{X \geq \frac{n}{2}+t\frac{\sqrt{n}}{2}\right\} \leq e^{-t^2/2}. \]

Rewriting $\xi=\frac{\sqrt{n}}{2}\cdot\frac{\sqrt{n}\epsilon}{2+\epsilon}$, and replacing $t$ with $\frac{\sqrt{n}\epsilon}{2+\epsilon}$, $X$ with $n/2+r^j$, the requirement for $(\epsilon,\delta)$-DP reduces to:
\begin{equation}\begin{split}\label{eq:binom-Chernoff}
	&e^{-t^2/2} \leq \frac{\delta}{2} \implies n \geq 2\left(\frac{2+\epsilon}{\epsilon}\right)^2\ln\left(\frac{2}{\delta}\right).
\end{split}\end{equation}
Therefore, Theorem~\ref{Theorem_BM} follows.
\end{proof}

It should be noted that when $\epsilon \ll 2$, $n \geq \frac{8}{\epsilon^2}\ln\left(\frac{2}{\delta}\right)$, providing a constant-factor improvement over the Binomial Mechanism in~\cite{dwork2006our}. Unlike Laplace or Gaussian Mechanism used in the original PATE~\cite{papernot2017semi,papernot2018scalable}, Binomial Mechanism avoids floating-point representation issues and enables efficient transmission, thus it can be seamlessly used with a cryptosystem. Furthermore, the stability of Binomial distribution facilitates noise distribution among multiple teachers.

\subsection{Discrete Gaussian Mechanism}
Discrete Gaussian (DG) Mechanism belongs to the category of Gaussian Mechanism~\cite{dwork2014algorithmic}, hence it satisfies all the properties of Gaussian Mechanism. However, in discrete Gaussian Mechanism, discrete Gaussian noise is sampled from a discrete Gaussian by the following Definition~\ref{def:Discrete-Gaussian} 

 \begin{definition}[Discrete Gaussian] \label{def:Discrete-Gaussian}
 The pmf of the discrete Gaussian is proportional to the pdf of its continuous version. For any $x \in \mathbb{Z}$, the pmf of discrete Gaussian is defined as:
 \begin{equation}\begin{split}
 P(X=x) \propto \displaystyle f_{X}(x) 
 &= \mathcal {N}(x;\mu _{X},\sigma_{X}^{2})\nonumber \\
 &= \frac {1}{{\sqrt {2\pi }}\sigma _{X}}\exp\left(-{\frac {(x-\mu
 _{X})^{2}}{2\sigma _{X}^{2}}}\right).\label{equ:DG}
 \end{split}\end{equation}
 \end{definition}

 \begin{corollary}[Stability of Discrete Gaussian]
 \label{corollary:stability_DG}
 The sum of independent discrete Gaussian distributed random variables still follows a discrete Gaussian distribution.
 \end{corollary}

Discrete Gaussian is a stable distribution as its continuous version (a sum of discrete Gaussian r.v.'s is still discrete Gaussian). The detailed proof for the stability of discrete Gaussian distribution can be referred to the Corollary A.1. in~\cite{lyu2018privacy_thesis}.
Like Binomial distribution, the stability of discrete Gaussian is ideal for realising DDP: analysis of overall privacy is made possible by analysing individuals, while also supporting fault tolerance if some individuals are compromised. Therefore, we utilize the stability of discrete Gaussian to distribute the noise generation among parties. To determine how much optimal noise should be added, we adopt an analytic Gaussian Mechanism as in Theorem~\ref{theorem_Analytic_Gaussian}, which eliminates the constraint $\epsilon <1$ in the classical Gaussian Mechanism and removes at least a third of the variance of the noise compared to the classical Gaussian Mechanism, thus delivering better utility~\cite{balle2018improving}.

\begin{theorem}
\label{theorem_Analytic_Gaussian}
Let $f: \mathbb{X} \rightarrow \mathcal{R}^d$ be a function with global $L_2$ sensitivity $\Delta$. For any $\epsilon \geq 0$ and $\delta \in [0,1]$, the Analytic Gaussian Mechanism $M(x) = f(x) + Z$ with $Z \sim N (0,\sigma^2 I)$ is $(\epsilon, \delta)$-DP if and only if

\begin{equation}\label{eq:Analytic_Gaussian}
\Phi \left(\frac{\Delta}{2\sigma}-\frac{\epsilon\sigma}{\Delta}\right) - e^\epsilon \Phi \left(-\frac{\Delta}{2\sigma}-\frac{\epsilon\sigma}{\Delta}\right) \leq \delta.
\end{equation}
\end{theorem}

In order to obtain $(\epsilon, \delta)$-DP for a function $f$ with global $L_2$ sensitivity $\Delta$, it is enough to add Gaussian noise with variance $\sigma^2$ satisfying Eq.~\ref{eq:Analytic_Gaussian}. We therefore distribute the discrete Gaussian noise at a level of $\sigma^2$ among parties.

\section{Distributed Privacy-preserving Prediction}
\label{sec:DPPP}
In this work, we study the applicability of DPPP to horizontally partitioned databases, where multiple parties each owns different groups of individuals with similar features. For example, different hospitals, each holding the same kind of information for different patients, can collaboratively perform statistical analyses of the union of their patients, while ensuring privacy for each patient. Consequently, hereafter, instead of training a centralized model to solve the task associated with the whole database $\mat{D} \in \mathbb{R}^{|D|\times d}, \vec{Y} \in \mathbb{R}^{|D|}$, the whole database $\mat{D}$ is partitioned into $N$ disjoint subsets, $\{\mat{D}_1,\mat{D}_2,...\mat{D}_N\}$ that are held by $N$ parties who are unwilling to make their training data, model parameters or model predictions public or share them with others. Here $|D|$ refers to the total number of training records in $\mat{D}$, $\mat{D}_i$ and $\vec{Y}_i$ represent party $i$'s training data and labels respectively. Individual models are trained separately on each subset $\{\mat{D}_i, \vec{Y}_i\}$.

In the case of prediction for any test point $\vec{x}$, each party applies $f(\vec{x},\mat{D}_i)=\vec{y}_i$, \ie a prediction function that returns the prediction $\vec{y}_i \in \{0, 1\}^c$ with the position of value 1 corresponding to the predicted class. The aggregate of multiple predictions becomes: $\sum_i f(\vec{x},\mat{D}_i) = \sum_i \vec{y}_i$, where $\vec{y}_i$ is the one-hot prediction vector produced by teacher $i$'s local model built on individual training data $\mat{D}_i$, hence the aggregate for each class equals to the sum of $N$ scalars. In \emph{distributed privacy-preserving prediction} (DPPP), the goal is to privately release the aggregated prediction, \ie noisy sum: ${\textstyle\sum}_i(\vec{y}_i+\vec{r}_i)$, where $\vec{r}_i=(r_i^1,\cdots,r_i^c)$, $r_i^j=z-mp$, $z \sim B(m,p), m=n/N$ and $p=1/2$ for BM, or $r_i^j \sim DG(0,\sigma/\sqrt{N})$ where $\sigma$ satisfies Eq.~\ref{eq:Analytic_Gaussian} for DGM. Our framework aims to deliver the differentially private aggregated prediction that is close to the desired aggregate $\sum_i \vec{y}_i$, while providing privacy guarantee. 

\begin{figure}[!htp]
\centering
\includegraphics[scale=0.2]{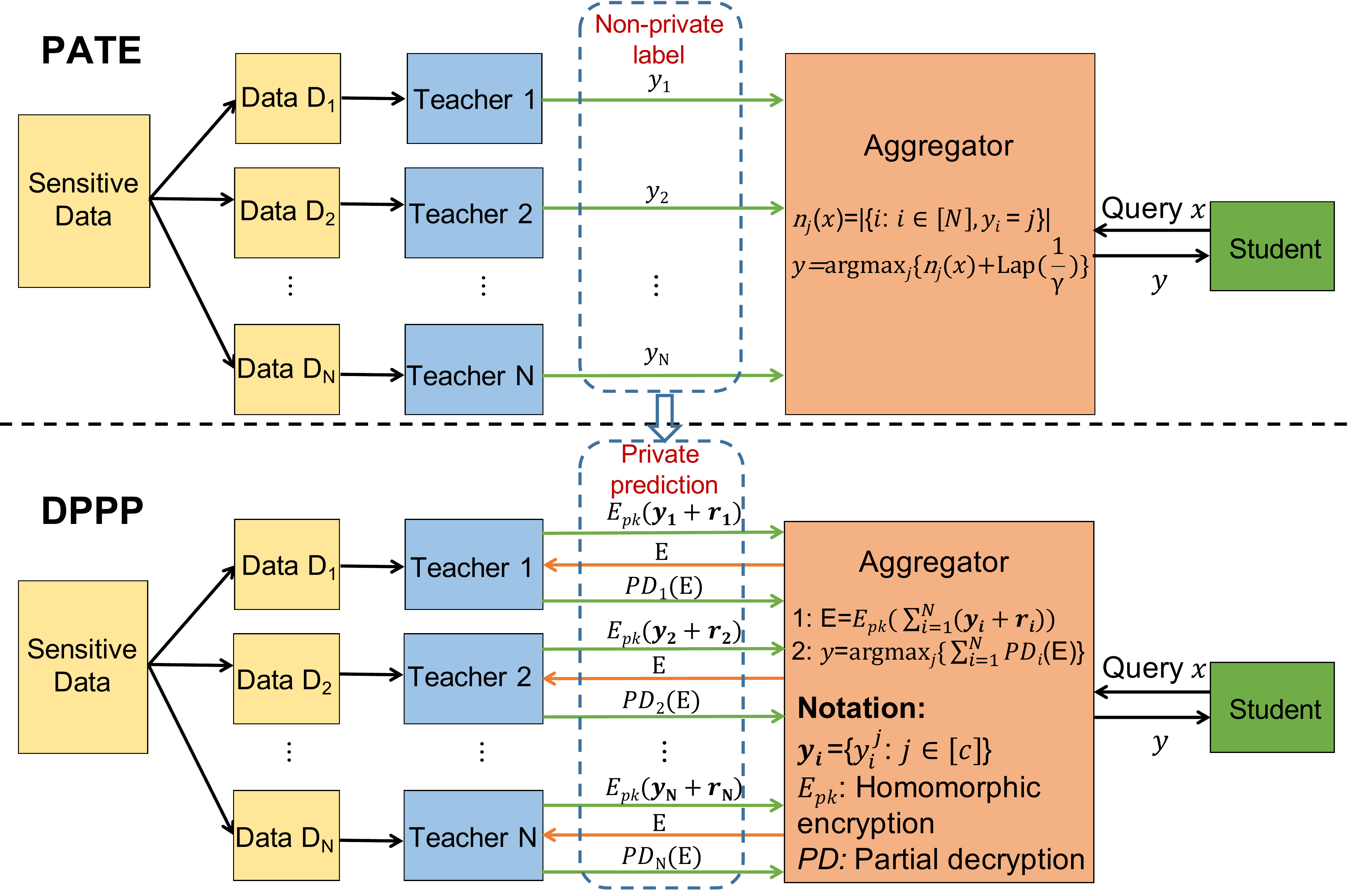}
\caption{Overview of the PATE and DPPP.}
\label{fig:case1}
\end{figure}
To realize this goal, we take inspiration from PATE framework. As illustrated in Fig.~\ref{fig:case1}, both PATE and DPPP first train an ensemble of teachers on disjoint subsets of the sensitive data, then the aggregator aggregates their outputs. However, the main difference between PATE and our DPPP is that the aggregator is trusted in PATE, so teachers directly share their non-private labels with the aggregator, while we improve PATE by eliminating the trust on the aggregator in DPPP. In particular, the total amount of noise required to guarantee $(\epsilon,\delta)$-DP of the aggregated prediction is distributed among teachers by using DDP: each teacher $i$ adds a share of noise $\vec{r}_i$ to its local prediction $\vec{y}_i$. The noise shares are chosen such that $\sum_{i=1}^N \vec{r}_i=\vec{r}$ is sufficient to ensure $(\epsilon,\delta)$-DP of the aggregated prediction, but $\vec{r}_i$ alone is not sufficient to ensure $(\epsilon,\delta)$-DP of local prediction, thus $\vec{y}_i+\vec{r}_i$ cannot be directly released to the aggregator. Therefore, it necessitates the help of cryptographic techniques to maintain utility and ensure aggregator obliviousness, as evidenced in~\cite{rastogi2010differentially,acs2011have,shi2011privacy}. Hence, we combine DDP with a distributed cryptosystem to achieve this goal. 
As shown in Fig.~\ref{fig:case1}, in DPPP, each teacher $i$ first computes the encryption of $\vec{y}_i+\vec{r}_i$ before sending it to the aggregator. Due to the additive homomorphic property in Eq.~\eqref{eq:add_hor}, the aggregator can compute the encryption of the noisy sum of all the local predictions as E $= E(\textstyle\sum_{i=1}^N (\vec{y}_i+\vec{r}_i))$.
This aggregated encryption E is then sent back to all teachers to compute partial decryptions. Finally, the partial decryptions $PD_i$(E) are forwarded to the aggregator, who combines all the partial decryptions to get the final decryption, \ie the aggregated prediction. 

\begin{theorem}
\label{theorem_aggregate}
Suppose $D,D'$ are neighboring databases that differ by one record, then each coordinate of the aggregated prediction, as given by $\textstyle\sum_{i=1}^N y_i^j$, differ by at most $1$.
Let $\mathcal{M}$ be the mechanism that reports $\textstyle\arg\max_j \left\{\sum_{i=1}^N \left(y_i^j+r_i^j\right)\right\}$. Then $\mathcal{M}$ satisfies $(\epsilon,\delta)$-DP, provided $r_i^j=B(m,1/2)-m/2, m=n/N, n \geq 2\left(\frac{2+\epsilon}{\epsilon}\right)^2\ln\left(\frac{2}{\delta}\right)$ or $r_i^j \sim DG(0,\sigma/\sqrt{N})$ where $\sigma$ satisfies Eq.~\ref{eq:Analytic_Gaussian}.
\end{theorem}

As stated in Theorem~\ref{theorem_aggregate}, the aggregated prediction is $(\epsilon,\delta)$ differentially private, the privacy guarantee stems from the aggregation of teacher ensemble. If teacher $i$ assigns class $j$ to input $\vec{x}$, then $y_i^j$ equals 1 while all the other elements are all 0's. For any test point $\vec{x}$, independent noise share $r_i^j$ is added to each teacher's prediction for each class $y_i^j$. Hence, the aggregated prediction for class $j$ is equivalent to $\textstyle\sum_{i=1}^N \left(y_i^j+r_i^j\right)$, and the predicted class equals:
\begin{equation}\label{eq:BM}
\arg\max_j \left\{\textstyle\sum_{i=1}^N \left(y_i^j+r_i^j\right)\right\},
\end{equation}
where $r_i^j=z-m/2, z \sim B(m,1/2), m=n/N$, here $n$ is the total number of tosses in Theorem~\ref{Theorem_BM} for BM, or $r_i^j \sim DG(0,\sigma/\sqrt{N})$ where $\sigma$ satisfies Eq.~\ref{eq:Analytic_Gaussian} for DGM. When there is a strong consensus among $N$ teachers, the label they almost all agree on (maximum of the aggregated prediction) does not depend on any particular teacher. Overall, DPPP provides a differentially private API: the privacy cost of each aggregated prediction made by the teacher ensemble is known. Semi-supervised learning can be further used to train a student model given a limited set of labels from the aggregation mechanisms~\cite{papernot2017semi,papernot2018scalable}.

\section{Distributed Cryptosystem}
\label{sec:cryptosystem}
As part of DPPP, based on the threshold Paillier cryptosystem~\cite{cramer2001multiparty}, we design a secure aggregation protocol in Protocol~\ref{protocol:sec_aggregate}, which can calculate the summation of teachers' local predictions without disclosing any of them. As we can see, the protocol mainly executes in two phases. In the first phase, the aggregator aggregates the encrypted noisy predictions as $E_{pk}(\sum_{i=1}^{N} \hat{\vec{y}}_i)$. Then in the second phase, a distributed decryption process is run to recover the aggregated noisy predictions $\textstyle\sum_{i=1}^{N} \hat{\vec{y}}_i$. In this protocol, what the aggregator received from all teachers are the encrypted noisy predictions and partial decryptions. Moreover, all the calculations on the aggregator are conducted on the encrypted data. What the aggregator can know is only the summation of all teachers' noisy predictions, based on which each teacher's local prediction $\vec{y}_i$ cannot be inferred, thereby providing aggregator obliviousness and significantly reducing privacy leakage. Note that the key generation needs to be done only once, hence secret-sharing protocols can be used for this purpose. 

\begin{algorithm}[htb]
\caption{Secure aggregation protocol}\label{protocol:sec_aggregate}
\begin{algorithmic}
\State
\begin{enumerate}
\item Each teacher $i$ encrypts its noisy prediction $\hat{\vec{y}}_i$ as $E_{pk}(\hat{\vec{y}}_i)$ as per Eq.~\eqref{eq:enc}, and sends it to the aggregator;
\item The aggregator computes $c=E_{pk}(\sum_{i=1}^{N} \hat{\vec{y}}_i)=\prod_{i=1}^{N} E_{pk}(\hat{\vec{y}}_i)$ based on Eq.~\eqref{eq:add_hor};
\item The aggregator sends $c$ to the randomly chosen $t$ teachers;
\item Each selected teacher $i$ calculates a partial decryption based on Eq.~\eqref{eq:partial_dec}, and sends it to the aggregator;
\item The aggregator combines all the partial decryptions to get the summation $\sum_{i=1}^{N} \hat{\vec{y}}_i$.
\end{enumerate}
\end{algorithmic}
\end{algorithm}

\emph{Fault Tolerance and Collusion}. In real system, one or more teachers might fail to respond or drop out the system at some point before the completion of the protocol for several different reasons. We also consider the threat of collusion among teachers, including the aggregator, through the trust parameter $t$ -- the minimum number of honest teachers. 
Our proposed DPPP can be made robust to the fault tolerance and collusion of less than 1/3 compromised teachers by adopting the following two solutions: (i) $(N,t)$-threshold decryption~\cite{damgaard2001generalisation} requires the cooperation of at least $t=\gamma N$ honest teachers for decryption, where $\gamma$ is the fraction of uncompromised teachers in Definition~\ref{def:ddp_database}. If $f$ teachers fail to send their partial decryptions, $(N,t)$-threshold decryption ensures that a decryption can still be computed as long as $f < N-t$; (ii) during the noise addition, for BM, each honest teacher sets its binomial noise as $z-m/2$, where $z \sim B(m,1/2), m=3n/2N$; for DGM, $r_i^j \sim DG(0,\sigma/\sqrt{2N/3})$ where $\sigma$ satisfies Eq.~\ref{eq:Analytic_Gaussian}, \ie leaving out $1/3$ teachers' randomness is still sufficient to ensure differential privacy. We remark that the number of honest parties $t$ could be set as per different applications, and $(N,t)$-threshold Paillier cryptosystem requires $2 \leq t \leq N$. Here the assumption of less than 1/3 compromised parties is often practical enough in most real scenarios. 

\section{Performance Evaluation}
\label{sec:Performance}
\textbf{Comparison frameworks.} To demonstrate the effectiveness of our DPPP, we compare it with the following frameworks: (1) \emph{Centralized non-private framework} requires all teachers to pool their training data into the aggregator to train a global model; 
(2) \emph{Distributed non-private framework} excludes both DP and cryptosystem, teachers directly share their local predictions with the aggregator; (3) \emph{Local differentially private (LDP) framework} excludes cryptosystem but requires each teacher to add the required level of noise to ensure ($\epsilon,\delta$)-LDP. The added noise is of the same level as the aggregated noise in DPPP, hence much more noise is added compared to the noise added in DPPP; (4) \emph{Standalone framework} allows teachers to individually train local models without any collaboration, and an end user directly sends a test query to one teacher, then local prediction is released under the guarantee of ($\epsilon,\delta$)-LDP; (5) PATE which relies on a trusted aggregator to add Laplace noise to the aggregated vote counts~\cite{papernot2017semi}. 

\begin{figure*}[!htp]
\centering
	\begin{subfigure}[b]{0.3\textwidth}
                \includegraphics[width=\linewidth,height=3.6cm]{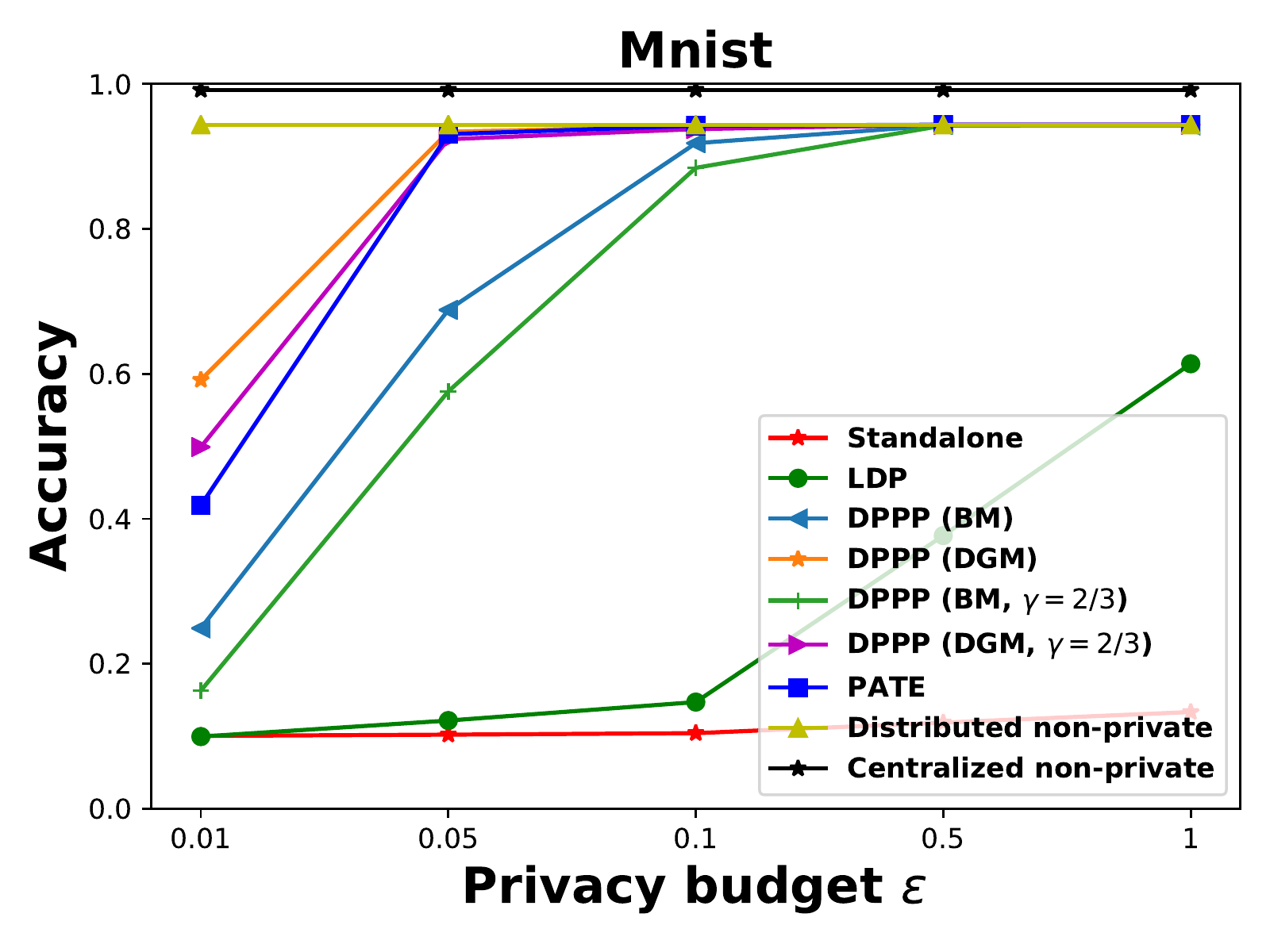}
				\label{fig:mnist_acc_BM}
        \end{subfigure}
        \begin{subfigure}[b]{0.3\textwidth}
                \includegraphics[width=\linewidth,height=3.6cm]{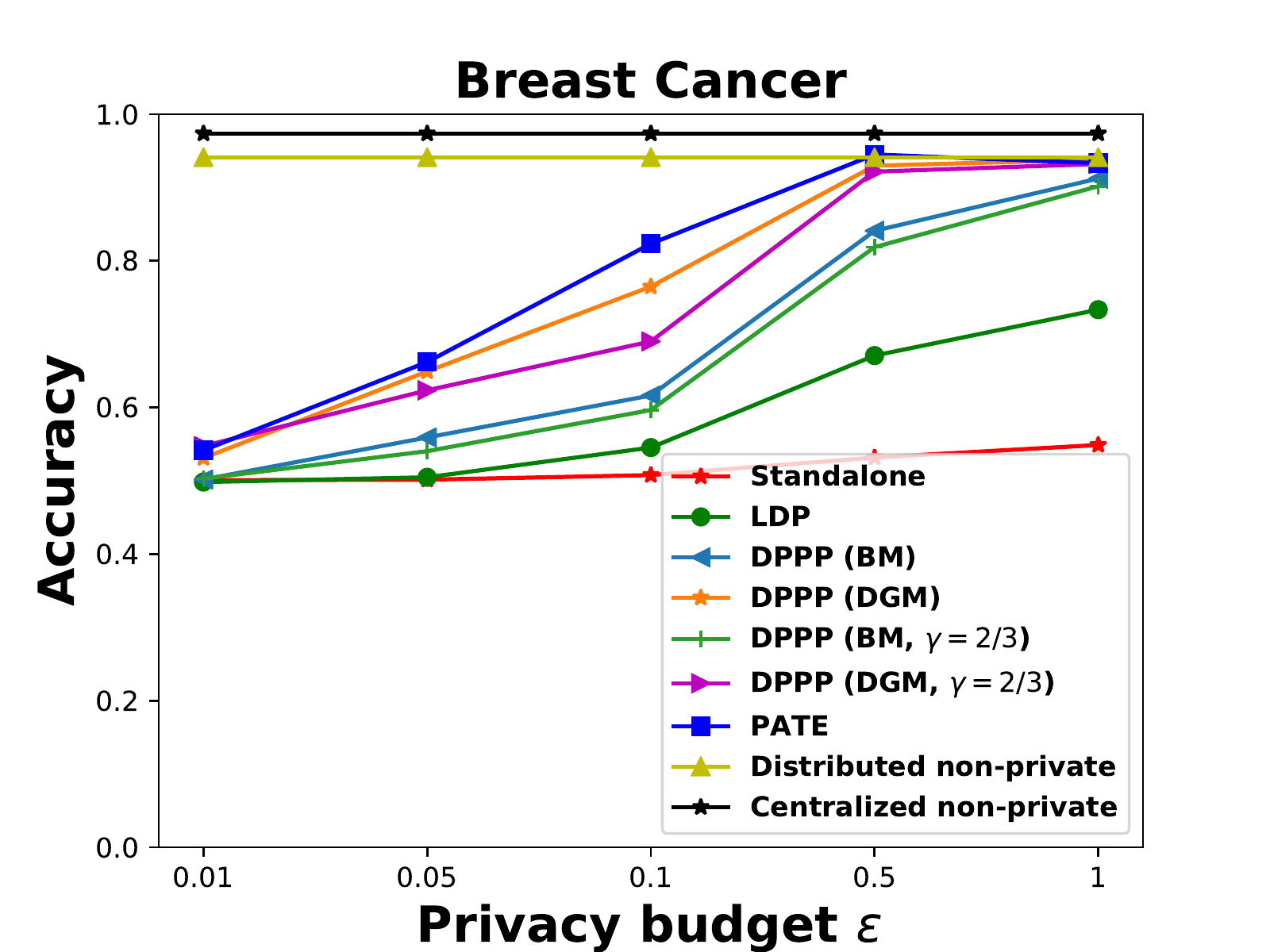}
				\label{fig:cancer_acc_BM}
        \end{subfigure}
         \begin{subfigure}[b]{0.3\textwidth}
                \includegraphics[width=\linewidth,height=3.6cm]{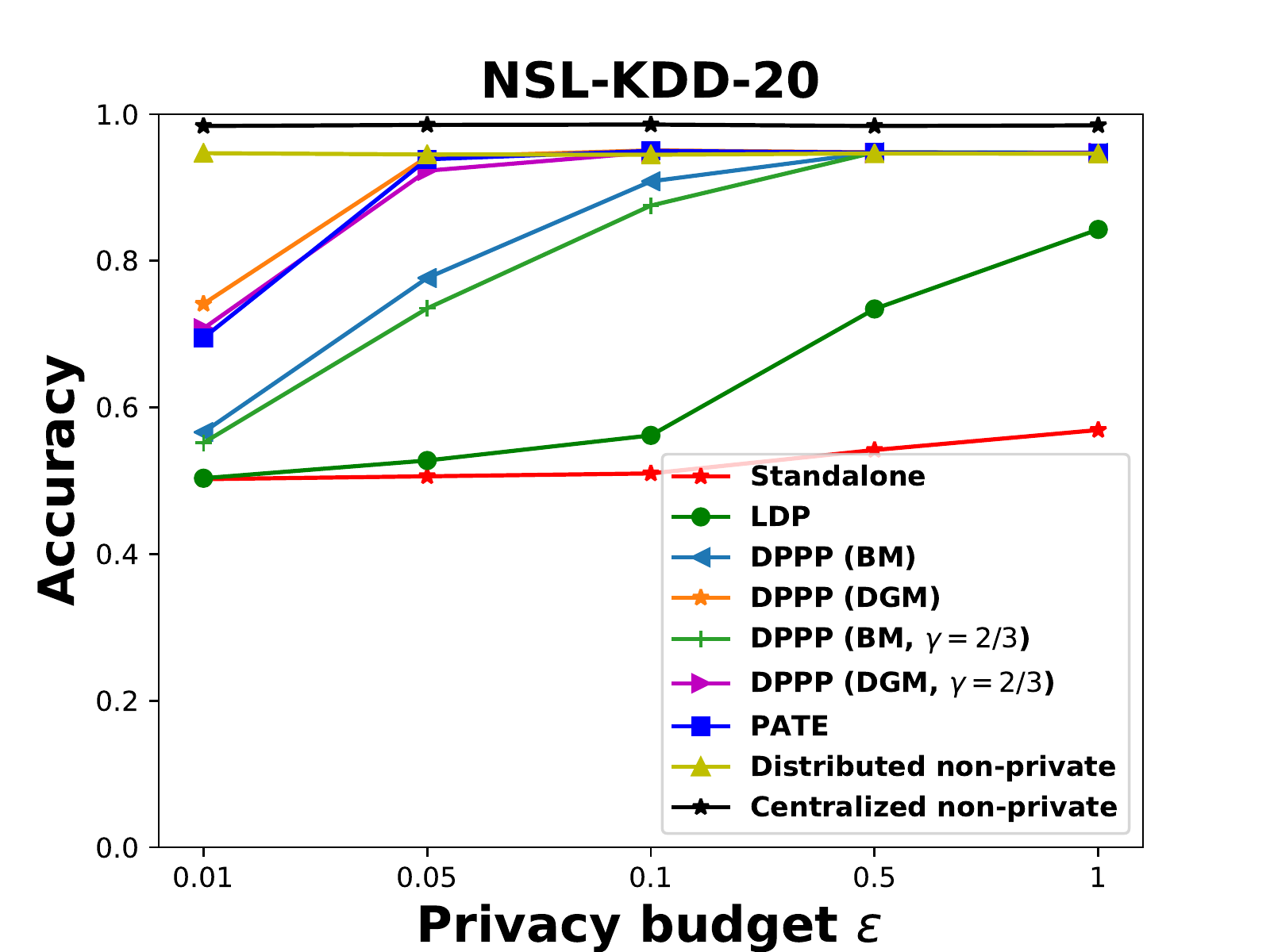}
				\label{fig:kdd20_acc_BM}
        \end{subfigure}
        \caption{Prediction accuracy of student queries for all datasets with varying $\epsilon$.}
\label{fig:acc_BM}
\end{figure*}

\begin{figure}[!htp]
\centering
	\begin{subfigure}[b]{0.24\textwidth}
                \includegraphics[width=\linewidth,height=3.4cm]{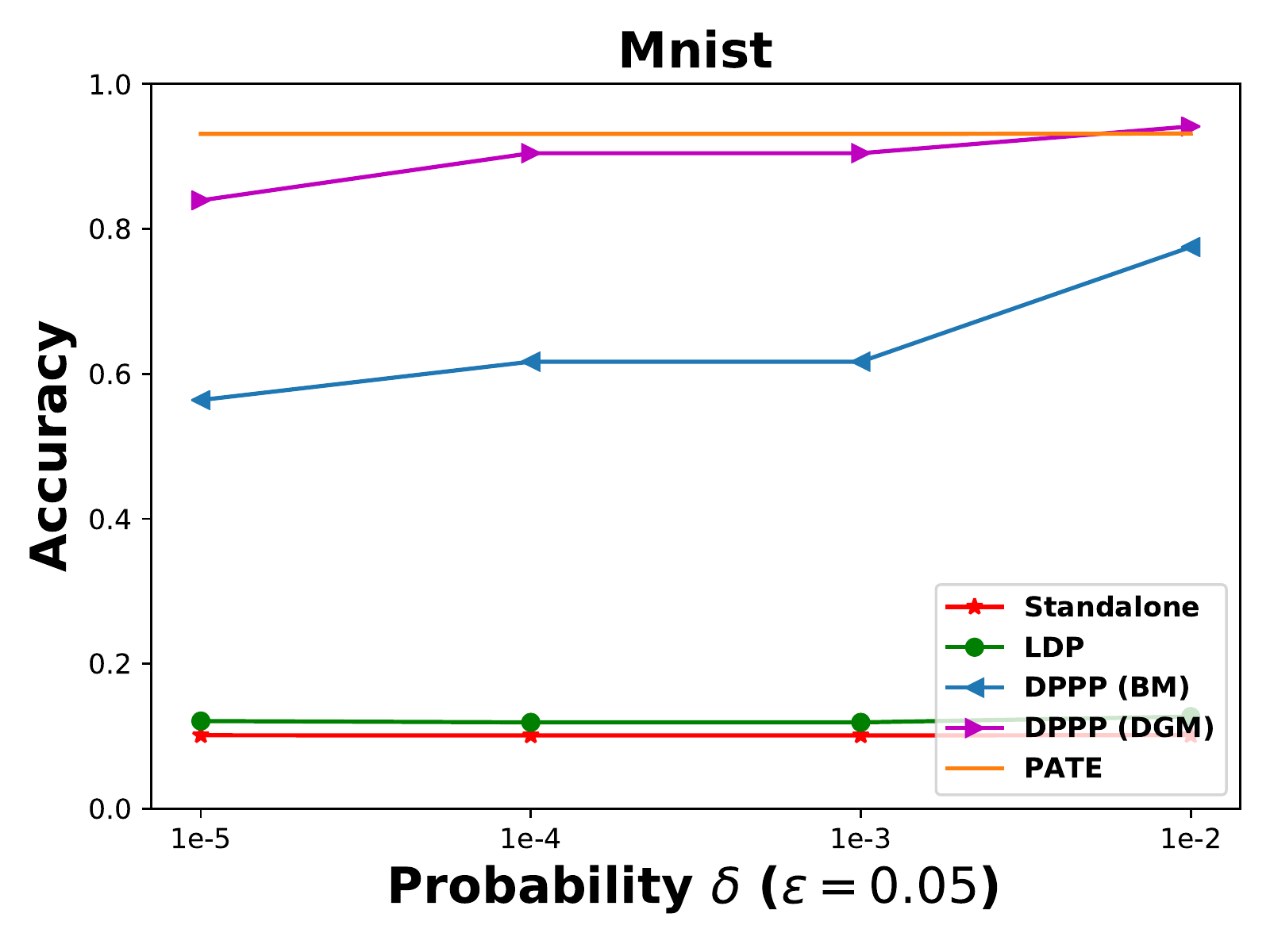}
				\label{fig:mnist_acc_BM_delta}
        \end{subfigure}
        \begin{subfigure}[b]{0.24\textwidth}
                \includegraphics[width=\linewidth,height=3.4cm]{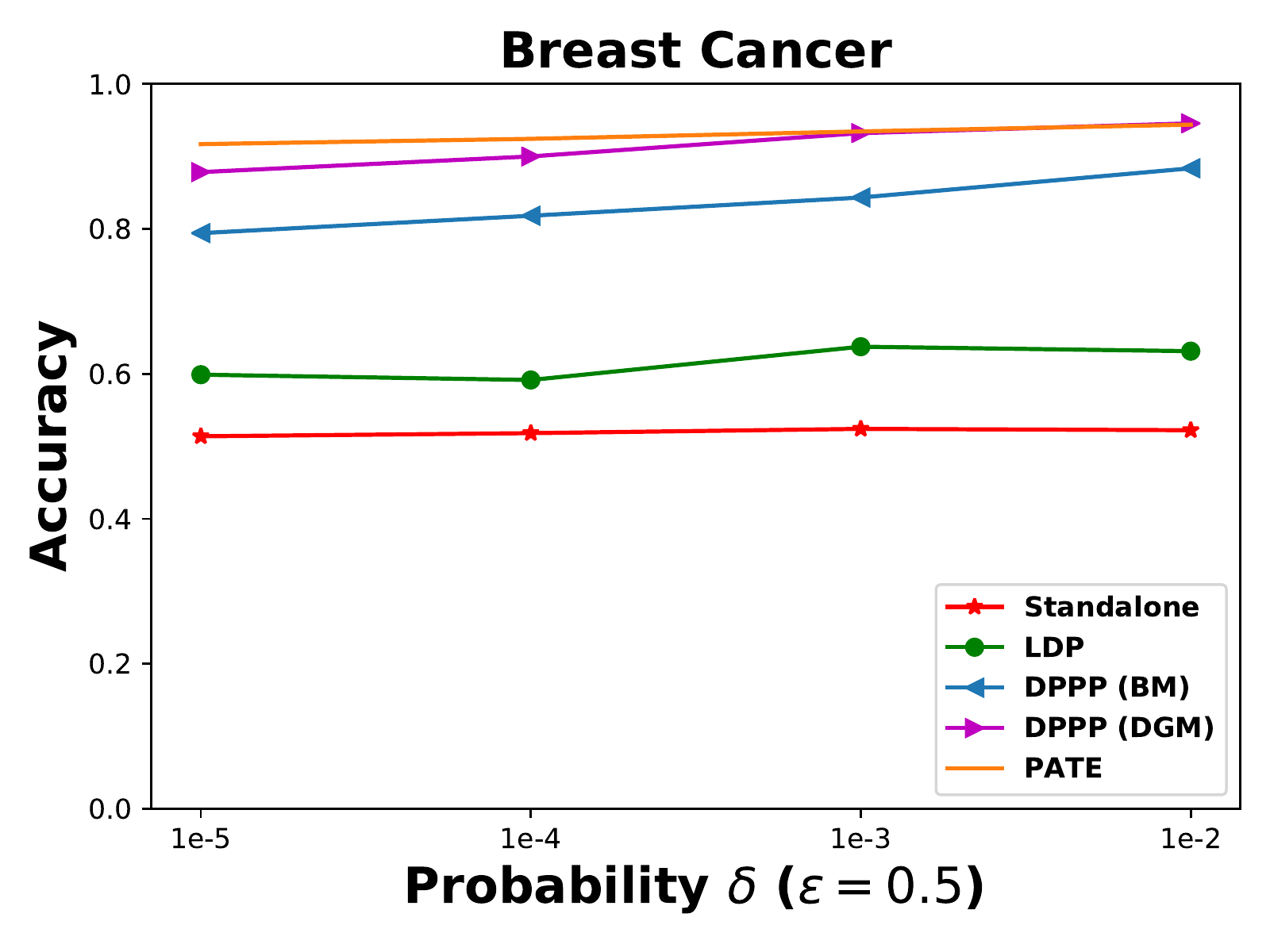}
				\label{fig:cancer_acc_BM_delta}
        \end{subfigure}
        \caption{Prediction accuracy of student queries with varying $\delta$.}
\label{fig:acc_BM_delta}
\end{figure}

\textbf{Datasets.} For fair comparison with PATE, we first adopt MNIST dataset, which consists of 60K training examples and 10K testing examples\footnote{\url{http://yann.lecun.com/exdb/mnist/}}. We also investigate the other two real-world datasets. One is Breast Cancer dataset~\footnote{\url{https://archive.ics.uci.edu/ml/datasets/Breast+Cancer+Wisconsin+(Diagnostic)}} 
that contains total 569 records with 32 features, each record is classified into two classes: malignant and benign. We randomly sampled 2/3 examples from the whole database as the training set, while the remaining 1/3 as the test set. The other one is NSL-KDD dataset used for intrusion detection\footnote{\url{https://www.unb.ca/cic/datasets/nsl.html}}, 
which contains total 125973 records with 41 features, each record is classified into two classes: anomaly and normal. We select a smaller subset called NSL-KDD-20 (with 20\% train and test data sampled from NSL-KDD). 

\textbf{Experimental Setup.} For MNIST, we follow~\cite{papernot2017semi} to use a convolution NN (CNN) model. 
Each teacher trains a convolutional network with two convolutional layers and one fully connected layer. For the other datasets, we use SVM model with RBF kernel. To simulate the situation in which each teacher constitutes only a limited subset of the whole database, all the training records are randomly distributed among multiple teachers such that each teacher receives nearly the same amount of records. Following the rationales provided in~\cite{papernot2017semi}, we \emph{empirically find appropriate values of $N$} for all the datasets by measuring the test accuracy of each teacher trained on one of the $N$ partitions of the whole training set, \ie we trained ensembles of 250, 100, 100 and 20 teachers for the MNIST, NSL-KDD-20 and Breast Cancer datasets respectively. How the number of teachers affects privacy cost can be referred to~\cite{papernot2017semi}. We run each experiment for 20 times and report the average result. $\gamma$ is set to be 1 or 2/3, \ie assuming no compromised teachers, or at most 1/3 compromised teachers.  
Since $\delta$ in BM and DGM are different from the classical DP mechanisms, we choose small values from $1e^{\{-5,-4,-3,-2\}}$. 

\textbf{Experimental Results.} We first fix $\delta=1e^{-3}$ and report the prediction accuracy of student queries under varying privacy budgets $\epsilon \in [0.01,1]$ in Fig.~\ref{fig:acc_BM}. In particular, the result of the standalone framework is derived by averaging over all teachers. As evidenced by Fig.~\ref{fig:acc_BM}, DPPP outperforms both the standalone and LDP frameworks, for all datasets. The centralized and distributed non-private frameworks achieve similar accuracy, indicating that the distributed non-private framework incurs minor accuracy degradation compared with the centralized non-private framework. Moreover, DPPP yields comparable accuracy to the distributed non-private framework when $\epsilon \geq 0.05$ for MNIST and NSL-KDD-20 datasets, which is also comparable to PATE where each query has a low privacy budget of $\epsilon = 0.05$~\cite{papernot2017semi}. We also notice that compared with other datasets, Breast Cancer dataset requires a higher value of $\epsilon$ to achieve comparable accuracy. One reason is the limited available data split among smaller number of teachers, while compensating for the introduced noise in Eq.~\eqref{eq:BM} requires large ensembles. We also observe that DPPP with DGM consistently outperforms DPPP with BM, and achieves comparable performance to the PATE. The extra noise introduced to guarantee privacy under 1/3 compromised teachers indeed slightly degrades accuracy, especially when $\epsilon$ is small. These findings provide empirical supports to guide the deployment of DPPP  framework.

We further show how the values of $\delta$ impact the performance by varying $\delta \in 1e^{\{-5,-4,-3,-2\}}$. We choose MNIST and Breast Cancer for illustration, and adopt the fixed $\epsilon=0.05$ for MNIST and $\epsilon=0.5$ for Breast Cancer. As can be observed in Fig.~\ref{fig:acc_BM_delta}, the effectiveness of our proposed DPPP persists, and for the fixed $\epsilon$, varying $\delta$ follows the similar trend as varying $\epsilon$. However, there is less difference with different $\delta$, which agrees with the findings reported in~\cite{abadi2016deep}. In contrast, $\epsilon$ has larger impact on accuracy.

\textbf{Computation and Communication Overhead.}
We use the typical 1024-bit key size and implement a $(N,\lfloor \frac{2}{3}N \rfloor)$-threshold Paillier cryptosystem using Paillier Encryption Toolbox\footnote{\url{http://cs.utdallas.edu/dspl/cgi-bin/pailliertoolbox/}}. 
The average computation time at a teacher is independent of the number of teachers and remains nearly constant. On the other hand, the time required by the aggregator might increase with the number of teachers, but this can be reduced by running the aggregator and teachers in parallel through MapReduce. In all cases, the computation overhead is quite small (within $\sim$ms), most of which is spent on cryptographic operations. For communication complexity, a Paillier ciphertext is estimated as 2048 bits (256 bytes). Therefore, the total communication cost between the aggregator and any teacher can be estimated as $256\times c\times 3=768\times c$ bytes, where $c$ is the number of classes and $3$ refers to three rounds of communication in Fig.~\ref{fig:case1}, which is also well within the realm of practicality as $c$ is usually small. 

\textbf{Discussion.} Different from the typical balanced and IID data distribution, in real practice, due to the differences in sensor quality, ambient noise, and skill level, the collected data by each teacher might be: (1) Unbalanced: due to the capabilities of different teachers, some teachers may have large training data, while others have little or no data. For substantially unbalanced data, most teachers have only a few examples, and a few teachers have a large number of examples. (2) Non-IID: the collected data by each teacher might not be representative of the population distribution. 

These two aspects are usually considered in federated learning, and might affect the accuracy of DPPP, especially when most teachers have few or extremely unrepresentative examples. To improve robustness to unbalanced and/or non-IID data distributions, current methods allow teachers to share locally trained model updates with the aggregator~\cite{mcmahan2016federated,mcmahan2018learning}, but giving the aggregator access to all teachers' updates clearly risks privacy leakage. To privately share individual model updates, Bonawitz \etal~\cite{bonawitz2017practical} proposed a secure aggregation protocol to securely aggregate local model updates as the weighted average to update the global model on the aggregator. However, this incurs both extra computation and communication costs.

\section{Conclusion and Future Work}
\label{sec:Conclusion}
We have presented a distributed privacy-preserving prediction framework, which enables multiple parties to collaboratively deliver more accurate predictions through an aggregation mechanism. Distributed differential privacy via Binomial Mechanism or Discrete Gaussian Mechanism and homomorphic encryption are combined to preserve individual privacy, maintain utility and ensure aggregator obliviousness. For the Binomial Mechanism, we offer tighter bounds than that in the previous works. Preliminary analysis and performance evaluation confirm the effectiveness of our framework. We plan to extend our framework to the unbalanced and non-IID data distribution. We also expect to extend our framework to various machine learning scenarios beyond classification. It is also important to investigate how to conduct privacy accounting for many subsequent queries by using different DDP mechanisms. 
\bibliographystyle{IEEEbib}
\bibliography{biblio}
\end{sloppypar} 
\end{document}